\documentclass[preprint,12pt]{article}

\expandafter\let\csname equation*\endcsname=\relax 
\expandafter\let\csname endequation*\endcsname=\relax 
\usepackage{amsmath,amssymb}
\usepackage{amsthm}
\usepackage{physics}

\usepackage{indentfirst}
\usepackage{mathrsfs}
\usepackage{graphicx}
\usepackage{xcolor}
\theoremstyle{plain}
\newtheorem{theorem}{Theorem}[section]
\newtheorem{lemma}{Lemma}

\theoremstyle{definition}

\theoremstyle{remark}
\newtheorem{remark}{Remark}[section]
\newcommand{\CO}[1]{\mathcal{O}\left(#1\right) }

\newcommand{\eps}{\epsilon}

\newcommand{\bbc}{\mathbb{C}}
\newcommand{\bbe}{\mathbb{E}}
\newcommand{\cE}{\mathcal{E}}

\newcommand{\ef}[1]{\exp\left( #1 \right)}

\begin{document}

\title{Some Error Analysis for the Quantum Phase Estimation Algorithms }
\author{Xiantao Li\\
  Department of Mathematics,\\
  The Pennsylvania State University, \\
  University Park, PA 16802, USA\\ 
{XLi@math.psu.edu}}
\date{}

\maketitle
\begin{abstract}
This paper is concerned with the phase estimation algorithm in quantum computing, especially the scenarios where (1) the input vector is not an eigenvector; (2) the unitary operator is approximated by  Trotter or Taylor expansion methods; (3)  random approximations are used for the unitary operator. We characterize the probability of computing the phase values in terms of the consistency error, including the residual error, Trotter splitting error, or statistical mean-square error. In the first two cases, we show that in order to obtain the phase value  with {error less or equal to  $2^{-n}$ } and probability at least $1-\eps$, the required number of qubits is  $ t \geq  n + \log \big(2 + \frac{\delta^2 }{2 \eps \Delta\!E^2 } \big).$ The parameter $\delta$ quantifies the error associated with the inexact eigenvector and/or the unitary operator, and $\Delta\! E$ characterizes the spectral gap, i.e., the separation from the rest of the phase values.  This analysis generalizes the standard result \cite{cleve1998quantum,nielsen2002quantum} by including these effects. More importantly, it shows that when $\delta<\Delta\!E$, the complexity remains the same. {  For the third case, we found a similar estimate, but the number of random steps has to be sufficiently large.}
\end{abstract}

\section{Introduction}

There has been a recent surge of development in quantum simulation algorithms for scientific computing problems, especially those that pose great challenges for classical computers. 
These developments also created demands for precise error estimates.  The phase estimation algorithm \cite{kitaev1995quantum}, which is an important application of the quantum Fourier transform, has been a  crucial element in many quantum algorithms \cite{cleve1998quantum,nielsen2002quantum}. Specific examples include Shor's algorithm  \cite{shor1999polynomial}, amplitude estimation \cite{brassard2002quantum},  quantum Metropolis sampling \cite{temme2011quantum}, 
 state preparation \cite{rendon2021effects}, the solution of large-scale linear system of equations 
 \cite{harrow2009quantum} and some nonlinear problems \cite{tansuwannont2019quantum}. In addition, it also has direct applications in quantum chemistry \cite{bian2019quantum,whitfield2011simulation,berni2015ab,poulin2018quantum,babbush2018encoding,bauman2020toward,lin2022heisenberg,wan2021randomized}. The algorithm has been included in various software packages \cite{wille2019ibm,casares2021t}.  \cite{berni2015ab} \cite{xiao2019continuous}

The phase estimation algorithm is often illustrated for a unitary operator $U$ using its eigenvector $\ket{\psi}$ as the input. By applying Hadamard gates, together with controlled-$U$ gates, the algorithm maps the binary bits of the phase variable to the computational basis, which can then be extracted using the inverse quantum Fourier transform \cite{nielsen2002quantum}. 

An insightful complexity analysis has been outlined in \cite{nielsen2002quantum}, which provides a lower bound for the number of required qubits, $ t \geq  n + \log \big(2 + \frac{ 1 }{ 2\eps  } \big)$,  in order to obtain $n$ correct bit values with success probability at least $1-\eps$. Later, a slightly improved estimate has been obtained in \cite{chappell2011precise}. In practice, however, the eigenvector $\ket{\psi}$ of $U$ is often unknown. Any approximation of $\ket{\psi}$  inevitably introduces error, which will be  subsequently propagated through the phase estimation method. In \cite{nielsen2002quantum}, this is incorporated into the complexity analysis in terms of the overlap with an eigenvector. However, such overlap might not be known in advance. Furthermore, of equal importance is quantifying the error in the computed phase value. In most problems in quantum chemistry, the unitary operator $U$ corresponding to a Hamiltonian $H$ can not be implemented exactly. Thus, the approximation of $U$ will also cause numerical error, and quantifying such error is important in understanding the effectiveness of the phase estimation method. Although the accur{acy in the approximation of $U$ has been extensively studied \cite{childs2019nearly,childs2021theory,childs2019faster,faehrmann2021randomizing,an2021time}, how such error affects the phase estimation algorithm remains an open issue. In addition, connecting the phase estimation error to the approximation error of $U$ will in turn help to determine the precision/complexity that is needed in the Hamiltonian simulation algorithms.  In the context of randomized algorithms \cite{campbell2019random,faehrmann2021randomizing,ouyang2020compilation}, the added stochasticity will further complicate the issue.   

The main purpose of this paper is to provide some a priori analysis of such error.
We will regard the above mentioned inexact treatment of $U$ and/or $\ket{\psi}$ as perturbations of the phase estimation procedure. For clarity, we separate these possible perturbations and analyze their effects, but also noting that in practice, these perturbations can occur simultaneously. 
One interesting observation by Hastings and coworkers \cite{hastings2014improving} is that an approximation of the unitary operator $U$ often corresponds to the dynamics of a modified Hamiltonian $H$, with a perturbation that is comparable to the local Trotter error. Therefore, the accuracy of the computed eigenvalues and phases should be connected to such error. 
{ The Trotter error, in terms of the operator norm, has been used as an upper bound on the perturbation of the eigenvalues of the unitaries \cite{campbell2021early,lin2022heisenberg,Kivlichan2020}, which is useful to assess the overall gate complexity.   } 
Our analysis confirms this observation. We express the error in terms of the residual error $\delta$. For an inexact eigenvector, the residual error, unlike the overlap with an exact eigenvalue, is computable. For approximate unitary operators, this error can be replaced by the local consistency error of a Trotter splitting scheme, which is proportional to $\tau^p$ with $\tau$ being the step size and $p$ being the order of accuracy. Our analysis also highlights the importance of the spectral gap $\Delta\!E$: Phase values that are well separated from others are much easier to compute. By taking into account the residual error, we prove a new bound:  $ t \geq  n + \log \big(2 + \frac{\delta^2 }{2 \eps \Delta\!E^2 } \big).$

Meanwhile, for a random Trotter scheme, such as the QDRIDFT method \cite{campbell2019random}, and later extensions \cite{jin2021partially,ouyang2020compilation},  the analysis  becomes more subtle: the modified Hamiltonian can be viewed as a Monte Carlo sampling of the Hamiltonian terms \cite{jin2021partially}. Although statistically, the random application of $U$ is accurate, one realization of such algorithms can only give a good approximation with certain probability that depends on the number of random steps.     Finally, we point out that although we only consider the basic phase estimation algorithm, we expect that similar analysis can also be applied to improved methods, such as the statistical approach \cite{xiao2019continuous,moore2021statistical}, 
the Cosine tapering window approach \cite{rendon2021effects}, and  the block encoding approach \cite{daskin2018direct}.

Throughout the paper, we will use the following notations. $\| \bullet \|$ denotes a vector norm, as well as the induced matrix norm. For each $N\in \mathbb{N}$, $[N]=\{1,2,\cdots, N\}.$ For a matrix $A$, we use $\sigma(A)$ to denote the set of eigenvalues.
Since the eigenvalues appear as the phases, we let $\mu(A)$ be the set of phases and we restrict $\mu(A) \subset [0,1).$  
Furthermore, to compare two phase values $\omega_1$ and $\omega_2$, we follow the notation in \cite{mosca1999quantum,brassard2002quantum} and define 
\begin{equation}
    \text{dist}(\omega_1, \omega_2):= \min_{z\in \mathbb{Z}} \abs{\omega_1 - \omega_2 + z}. 
\end{equation}
In particular, this distance is always bounded by $1/2.$

It will also be useful to measure the distance between a phase value and a subset of $\mu(A)$. Toward this end, we extend the previous definition: For a set  $S \subset [0,1)$, and $\omega \in [0,1)$;
\begin{equation}
    \textrm{dist}(\omega, S ) := \min_{s \in S} \text{dist}(\omega, s). 
\end{equation}

Another useful tool is the Jordan's inequality,
\begin{equation}\label{jordan}
    \sin (x) \ge \frac{2}{\pi} x, \quad \forall \; x \in [0, \frac{\pi}2].
\end{equation}

\section{The Phase Estimation Method}
We consider the basic phase estimation algorithm \cite{cleve1998quantum,nielsen2002quantum}.
This method starts with a quantum register with $t$ qubits, together with an input quantum state $\ket{\psi}.$ Let 
\begin{equation}
U=\ef{iH},
\end{equation}
be the unitary operator associated with the Hamiltonian dynamics, the goal is to extract an eigenvalue of $H$. $H\in \bbc^{d\times d}$; $d$ represents the dimension of the quantum system.  An ideal scenario is when $\ket{\psi}$ is an eigenvector  associated with a phase value $ \cE \in \mu(H)$, 
\begin{equation}\label{eq: eigp}
    H \ket{\psi} = 2\pi \cE \ket{\psi}, \; \text{and} \;  U \ket{\psi} =  e^{ 2 \pi i  \cE} \ket{\psi}.
\end{equation}
$\cE \in [0,1)$ will be referred to as an eigenphase of the unitary operator $U$. 

The standard approach consists of 
applying the Hadamard gates, followed by an inverse quantum Fourier transform. Following the notations in \cite{daskin2018direct}, we divide the procedure into the following steps.
\begin{align}
 \ket{\phi_0} = & \ket{0^{\otimes t}} \ket{\psi}, \label{eq: qpe0} \\
 \ket{\phi_1}= & \dfrac{1}{2^{t/2}}  \bigotimes_{k=0}^{t-1} \left( \ket{0} + \ket{1}  \right)  \ket{\psi}, \\ 
 \ket{\phi_2}= & \dfrac{1}{2^{t/2}}  \bigotimes_{k=0}^{t-1} \left( \ket{0} + e^{2\pi i  2^k \cE} \ket{1}  \right)  \ket{\psi},\\
  \ket{\phi_3} = &  \dfrac{1}{2^{t}} \sum_{k=0}^{2^t-1}\sum_{j=0}^{2^t-1}  e^{2\pi i  k (\cE - j/2^t)} \ket{j}  \ket{\psi}. \label{eq: qpe3}
\end{align}
 
 The result from the second step can also be conveniently written in the computational basis as,
 \begin{equation}
 \ket{\phi_2}=   \dfrac{1}{2^{t/2}}  \sum_{k=0}^{2^t-1}  e^{2\pi i  k \cE} \ket{k}   \ket{\psi}.
\end{equation} 
Consequently $\ket{\phi_3}$ follows from an inverse Fourier transform.

An analysis of the complexity has been shown in \cite{cleve1998quantum,nielsen2002quantum}. The first step is to designate $b \in [N]$ as an approximation of the eigenvalue, in the sense that
\begin{equation}
    0 \leq \nu:= \cE - b/2^t  \leq 2^{-t}.
\end{equation}

The principal idea in the analysis is to express the state at the last step as,
\begin{equation}
     \ket{\phi_3} =    \sum_{j=0}^{2^t-1}  \alpha_j  \ket{(j+b)\text{mod}2^t }  \ket{\psi}
\end{equation}
In particular, we have the coefficients given by \cite{mosca1999quantum},
\begin{equation}
    \abs{\alpha_j} = \frac{\abs{\sin 2^t \pi \beta_j } }{ 2^t \abs{\sin \pi \beta_j} }, \quad \beta_j= \nu - \frac{j}{2^t}.
\end{equation}
Let $m$ be the measured outcome. The probability of $m$ being close to $b$ can be bounded by considering the coeffients,
\begin{equation}
    \mathbb{P} \left( |m - b| > \ell \right) \le 
     \sum_{-2^{-t+1} < j < -\ell } \abs{\alpha_j}^2 
     + 
     \sum_{ \ell < j \leq  2^{t-1} } \abs{\alpha_j}^2.
\end{equation}
Here $\ell$ is an integer.  It has been shown in \cite{nielsen2002quantum} that,
\begin{equation}\label{eq: 15}
 \sum_{-2^{-t+1} < j < -\ell } \abs{\alpha_j}^2  + 
  \sum_{ 
  \ell < {j} \leq  2^{t-1} } \abs{\alpha_j}^2 \leq \frac{1}{2(\ell-1)}.
\end{equation}
Similar analysis, but in the context of amplitude amplification, has been done in \cite{brassard2002quantum}.

By setting $\ell =2^{t-n}$, the estimated phase will have accuracy within $2^{-n}$. In addition,  $\eps = \frac{1}{2(\ell-1)} $ on the right hand side of Eq. \eqref{eq: 15} can be viewed as a bound for the failure probability. Combining these two equations,  one obtains an estimate for the number of qubits $t$, given a tolerance $\eps$.  
\begin{theorem}[Nielsen and Chuang \cite{nielsen2002quantum} ]\label{thm1}
Assume that the initial state is $\ket{0^{\otimes t}}\ket{\psi}$ with $\ket{\psi}$ being the eigenvector of $H$ associated with the eigenvalue $2\pi \cE$. For any $n\in \mathbb{N}$ and $\eps > 0$, the phase value $\cE$ can be obtained with {error less or equal to $2^{-n}$ and} with probability at least $1-\eps$, if one chooses the number of qubits,
\begin{equation}
    t \geq n + \log \left(2+\frac{1}{2\eps} \right).
\end{equation}
\end{theorem} 

One approach to improve the accuracy is to use amplitude amplification, where one works with a function $f: \{0,1, \cdots, 2^{t}-1\} \to \{0,1\}.$ For the current problem, one can define $f(j)=1,$ if $j=b$ and zero otherwise. Then the majority algorithm in \cite{brassard2011optimal} can improve the probability to $1-1/m$ by only repeating the algorithm $\CO{\log m}$ times.

\bigskip

The first problem that comes up is that the input state $\ket{\psi}$ is not an eigenvector. To gain insight, we express $U$ in a spectral decomposition form,
\begin{align}\label{eq: U-SD}
    U= \sum_{m=1}^d e^{2\pi i \cE_m} \ketbra{\psi_m}.
\end{align}

After applying the controlled-$U$ gate, we find that,
\begin{align}
\ket{\phi_2}=& \dfrac{1}{2^{t/2}}   \sum_{k=0}^{2^t-1}   \sum_{m=1}^d  e^{2\pi i k \cE_m } \braket{\psi_m}{\psi} \ket{k}\ket{\psi_m},\\
\ket{\phi_3}= &  \dfrac{1}{2^{t}} \sum_{k=0}^{2^t-1} \; \sum_{j=0}^{2^t-1}  e^{-2\pi i  k   j/2^t} \ket{j} \sum_{m=1}^d  e^{2\pi i k \cE_m  } \braket{\psi_m}{\psi}  \ket{\psi_m}.\label{phi3}
\end{align}

We extend the analysis from the previous section, by writing $\ket{\phi_3}$ as follows,
\begin{equation}
  \ket{\phi_3} = \sum_{m=1}^d  \sum_{j=0}^{2^t-1}  \alpha_{j,m}  \ket{(j+b)\text{mod}2^n }  \ket{\psi_m},
\end{equation}
where we have defined the coefficients,
\begin{equation}
    \alpha_{j,m}= \sum_{k=0}^{2^t-1}  e^{2\pi i  k (\cE_m - (b+j) )  /2^t)} \braket{\psi_m}{\psi} 
\end{equation}
Therefore, the problem is reduced to estimating the magnitude of  $\alpha_{j,m}$.  Observing that this is also a geometric series, we rewrite them as,
\begin{equation}\label{eq: alphajm}
    \abs{\alpha_{j,m}} = \frac{\abs{\sin 2^t \pi \theta_{j,m} } }{ 2^n \abs{\sin \pi \theta_{j,m} } } \abs{\braket{\psi_m}{\psi}}, \quad \theta_{j,m}:= \cE_m - b - \frac{j}{2^t}.
\end{equation}

Intuitively, when there is significant overlap between $\ket{\psi}$ and 
$\ket{\psi_m}$ for some $m$, then the phase estimation method will obtain $\cE_m$ with high probability \cite{nielsen2002quantum}. Unfortunately, such overlap is not accessible unless one diagonalizes the matrix $H$, which is generally not practical. 
For a Hermitian matrix, a standard approach to examine whether a vector is close to being an eigenvector is to evaluate the residual error \cite{saad2011numerical}. 

\begin{lemma}[\cite{saad2011numerical}]\label{lem-r}
 Let $H \in \bbc^{d\times d}$ be a Hermitian operator. Let  $a\in \mathbb{R}$, and $\ket{\psi} \in \bbc^d$  with norm 1. Suppose that,
 \[   \norm{ H\ket{\psi} - a\ket{\psi} } < \delta, \]
for some $\delta>0.$ Then there exists an eigenvalue  $E \in \sigma(H)$ with eigenvector $\ket{\chi}$, such that $|E-a| < \delta.$ In addition, let 
\begin{equation}\label{eq: dE}
    \Delta E= 
    \min_{\overset{\lambda\in \sigma(H)}{\lambda \neq E} } 
 \norm{a -  \lambda}.
\end{equation}
Then  the overlap between the approximate and the exact eigenvectors satisfies the following bound,
\begin{equation}\label{eq: ovlp1}
    \abs{ \braket{\psi}{\chi} }^2 \geq 1 - 
    \frac{\delta^2}{\Delta\!E^2}.
 \end{equation}
 
\end{lemma} 

It is worthwhile to point out that $\Delta E$ in Eq. \eqref{eq: dE} characterizes the separation of $E$ from the rest of the eigenvalues. In classical algorithms \cite{saad2011numerical}, a small separation will typically cause large error in solving the eigenvalue problem, { especially when eigenvectors are desired.  The role of the spectral gap is also observed in the recent work on quantum algorithms for ground state preparations \cite{lin2022heisenberg,lin2020near}.   }

The lower bound in Eq. \eqref{eq: ovlp1} also implies, in terms of the spectral decomposition \eqref{eq: U-SD}, that
\begin{equation}\label{eq: summ}
    \sum_{\overset{1\leq m'\leq d}{m'\neq m}} \abs{\braket{\psi_{m'}}{\psi}}^2 \leq \frac{\delta^2}{\Delta E^2}.
\end{equation}

\begin{theorem}\label{eq: thm2}
   Under the condition that  $\norm{ H\ket{\psi} - a\ket{\psi} } < \delta$ for some $a \in [0,2\pi)$, and some sufficiently small $\delta$. 
   The quantum algorithm \eqref{eq: qpe0} -- \eqref{eq: qpe3} returns an estimation within success probability $1-\eps$ by using the number of qubits given by,
   \begin{equation}
       t \geq  n + \log \big(2 + \frac{\delta^2 }{2 \eps \Delta E^2 } \big). 
   \end{equation}
\end{theorem}

This estimate generalizes the one obtain in Theorem \ref{thm1} by incorporating the residual error and the eigenphase separation.

\begin{proof}

Following Lemma \ref{lem-r}, let $\cE_m \in \mu(H)$ such that $\abs{2\pi \cE_m - a} < \delta.$  This index $m$ is held fixed hereafter. 
 Starting with Eq. \eqref{eq: alphajm}, we choose $b$ so that $ \delta_m := \abs{2^{-t} b-\cE_m},$  and $ \delta_m  \le 2^{-t}.$
  We first observe that 
 \begin{align*}
    &    |\alpha_{j,m'}|^2 \leq \frac{1}{4 (2^n \delta_{m'} - j)^2 } \abs{\braket{\psi_{m'}}{\psi}}^2 \\
     \Longrightarrow&   
     \sum_{m'\neq m}   |\alpha_{j,{m'}}|^2  
       \leq  \frac{1}{4 (2^n \delta_{m} - j)^2 } \frac{\delta^2}{\Delta E^2},\\
        \Longrightarrow& 
        \mathbb{P}\left(\abs{m'-b} > \ell \right) \leq
          \frac{\delta^2}{2 \Delta E^2 (\ell-1) }.
 \end{align*}
 In the second step, we have used the inequality \eqref{eq: summ}.      The rest follows from the proof of Theorem \ref{thm1}.

\end{proof}

\begin{remark}
  One interesting scenario is when several phase values are clustered around $\cE_m$ into a very small interval. In this case $\Delta\!E$ is very small and the estimate given above is rather pessimistic. On the other hand, since we are only computing the phase values, rather than the eigenvectors, we can extend this analysis by including the multiple phase values: {When the cluster size is less than $2^{-n}$, all the values in the cluster are good approximations. } From this perspective, as indicated in Figure \ref{default}, $\Delta E$ should be regarded as the distance between the small cluster and the rest of the phase values. 
\end{remark}

{
\begin{remark}
The residual error can be computed in advance, and it calculation can be done fairly efficiently when the Hamiltonian matrix is sparse. On the other hand, the calculation of the spectral gap is not as straightforward. One possible approach to probe the eigenvalue distribution is by computing the density of states (DoS) \cite{Lin2016,saad2011numerical}.  

\end{remark}
}

\begin{figure}[htbp]
\begin{center}
\includegraphics[scale=0.76]{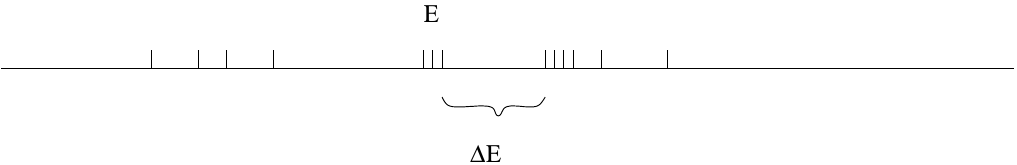}
\caption{An illustration of the phase separation: $\Delta E$ represents the separation of a small cluster of phase values around $E$ from the rest of the phase values. }
\label{default}
\end{center}
\end{figure}

\section{Inexact Unitary Operators}

Another important, but common, scenario is when the unitary operator can not be implemented exactly. For instance, the total Hamiltonian can be written as the sum of  $L$ local Hamiltonians,
\begin{equation}
    H = \sum_{\ell=1}^L h_\ell,
\end{equation}
and one can assume that  the unitary operator for each  Hamiltonian term $h_\ell$ can be implemented exactly  \cite{Lloyd1996}. {Such algorithms are known as Hamiltonian simulations, and they have been a rapidly growing area of interest \cite{childs2021theory,poulin_trotter_2014,babbush_chemical_2015,Berry2015,Low2017,Low2018,Low2019,campbell2019random,Berry2020,ouyang2020compilation,Chen2021,Su2021,An2021,an2021time,Watkins2022}.
 } Among the many  methods that have been developed to approximate the unitary operator $U=\exp(iH)$, perhaps the simplest method is the Trotter splitting. One first regards the computation of $U$ as a time integration, and divides the time interval into $r$ segments, with step size 
\begin{equation}
0< \tau:=1/r \ll 1.
\end{equation}
For each time integration, the unitary operator is approximated by,
\begin{equation}\label{eq: V-split}
     \ef{i \tau H} \approx \prod_{j} \ef{i \tau c_j h_j},  
\end{equation}
where the product is composed of terms $\ef{i \tau c_j h_j}$ with each $j \in [L]$ and each $c_j$ being a real coefficient.   One such example is the second-order symmetric splitting,
\begin{equation}\label{strang}
     \ef{i \tau H} \approx \ef{i \frac{\tau}2 h_1}
      \cdots \ef{i \frac{\tau}2 h_{L-1}} \ef{i \tau h_L}\ef{i \frac{\tau}2 h_{L-1}} \cdots \ef{i \frac{\tau}2 h_1} + \CO{\tau^3}.
\end{equation}

We denote the approximation in \eqref{eq: V-split}, applied repeatedly for $r$ steps, simply by $V$:
\begin{equation}
V=  \Big( \prod_{j} \ef{i \tau c_j h_j} \Big)^r.
\end{equation}

Since we are using an approximate unitary operator, we postulate the following standard local order condition \cite{deuflhard2002scientific},
\begin{equation}\label{eq: local-err}
    \norm{ U^{1/r} \ket{\psi} - V^{1/r} \ket{\psi} }\leq C \tau^{p+1},
\end{equation}
for some $p \in \mathbb{N}$ and for any normalized wave function $\ket{\psi} \in \bbc^d.$ The constant $C$ is independent of the step size $\tau$. By using the fact that both $U$ and $V$ are unitary, one can extend \eqref{eq: local-err} to,
\begin{equation}\label{eq: local-err-k}
    \norm{ U^{k/r} \ket{\psi} - V^{k/r} \ket{\psi} }\leq C k \tau^{p+1}, \quad \forall k \in \mathbb{N}.
\end{equation}
In particular, 
\[ \norm{ U \ket{\psi} - V \ket{\psi} }\leq C \tau^{p}.
\]

Let us express $V$ in a spectral decomposition form as well,
\begin{align}\label{eq: V-SD}
    V= \sum_{m=1}^d e^{2\pi i E_m} \ketbra{\chi_m}.
\end{align}
 
This draws a comparison of the spectral decomposition for the exact unitary \eqref{eq: U-SD} and the approximate one \eqref{eq: V-SD}. We first show that the error bound  \eqref{eq: local-err} also implies a bound for the eigenphase approximation. 
\begin{lemma}\label{lmm2}
Assume that the local error \eqref{eq: local-err} is sufficiently small. Let $\ket{\psi}$ be an eigenvector of $H$ with eigenvalue ${2\pi \cE}.$  Then we have the eigenvalue bound,
\begin{equation}\label{eq: mmin}
     \min_{1 \leq m \leq d} 
     \mathrm{dist}(\cE, E_m )  \leq \frac12 C \tau^{{p}}.
\end{equation}

Furthermore, the overlap follows the lower bound,
\begin{equation}\label{eq: ovlp2}
     \abs{\braket{\chi_m }{\psi} }^2 \geq 1 -  \Big( \frac{ C  \tau^{p}}{2\Delta E} \Big)^2.
\end{equation}
Here $m$ is the integer where the minimum in Eq. \eqref{eq: mmin} is achieved.  $\Delta E$ is defined as,
\begin{equation}\label{eq: de'}
    \Delta E = \min_{m'\ne m} \mathrm{dist}( E_{m'}, \cE ).
\end{equation}

\end{lemma}

\begin{proof}

We begin with the vectors, 
 \[ 
 \begin{aligned}
      V^\dagger U \ket{\psi} - \ket{\psi} & =
 \sum_{m=1}^d \big(e^{-2\pi i (E_m - \cE) } -1 \big)\ket{\chi_m} \braket{\chi_m }{\psi}, \\
 \Longrightarrow & \norm { U \ket{\psi} - V \ket{\psi} }^2 =  \sum_{m=1}^d  \abs{\braket{\chi_m }{\psi} }^2 \abs{ \sin^2 \pi (E_m - \cE) } \\
 & \geq   
\min_{1 \le m \le d} \abs{ \sin^2 \pi (E_m - \cE) } \geq 4 \min_{1 \le m \le d}  \mathrm{dist}(E_m, \cE)^2.  
 \end{aligned}
 \]
 The last step follows from the orthogonality of the eigenvectors $\ket{\chi_m}$ and the Jordan's inequality \eqref{jordan}. The error bound \eqref{eq: mmin} can now be obtained by using the approximation error bound \eqref{eq: local-err}.
 
Let us fix the index $m$ where the minimum is achieved.  The second part of the theorem can be checked by observing that,
 \begin{align*}
 \Big( C \tau^{{p}} \Big)^2 \geq & \norm{U \ket{\psi} -V \ket{\psi}}^2\\
   = & \norm{ V^\dagger U \ket{\psi} - \ket{\psi}}^2 \\
   \geq& \sum_{m' \neq m}  \abs{\braket{\chi_{m'} }{\psi} }^2 \abs{ \sin^2 \pi (E_{m'} - \cE) }  \\ \geq&  \sum_{m' \neq m}  4 \abs{\braket{\chi_{m'} }{\psi} }^2  \Delta E^2, \\
  \Longrightarrow & 
  \abs{\braket{\chi_m }{\psi} }^2 
  = 1 - \sum_{m' \neq m}  \abs{\braket{\chi_{m'} }{\psi} }^2 
  \geq 1 - \Big( \frac{C \tau^{p} }{2\Delta E} \Big)^2.
 \end{align*}

 
\end{proof}

In light of the similar lower bounds \eqref{eq: ovlp1} and \eqref{eq: ovlp2},  we conclude that: 
\begin{theorem}\label{thm3}
 Suppose that an approximate Hamiltonian method $V$ is implemented in the phase estimation algorithm. Under the condition that \eqref{eq: local-err} is small enough,  the quantum algorithm returns an estimation with {error less or equal to $2^{-n}$ and} with   success probability at least $1-\eps$, using at most $t$ qubits
  \begin{equation}\label{eq: t3}
       t > n + \log \big(2 + \frac{ \big(C \tau^{p+1}\big)^2 }{2  \Delta\!E^2 \eps} \big). 
   \end{equation}
Here $ \Delta\!E$ is defined in \eqref{eq: de'}.   
\end{theorem}

\begin{remark}{
If we simply combine \eqref{eq: ovlp1} and \eqref{eq: ovlp2}, we would get a $\tau^p$ term in  the estimate  \eqref{eq: t3}, rather than $\tau^{p+1}$. The slightly improved estimate is obtained by observing } that the  Trotter error \eqref{eq: local-err-k} did not play a direct role in the analysis. Instead, the error is represented by the overlap of the vectors \eqref{eq: ovlp2}. In the case when the approximation $V$ is composed of repeated one-step approximation, i.e.,
\begin{equation}
    V= ( V^{1/r} )^r,
\end{equation}
where $V^{1/r}$ corresponds to an approximation of $\ef{i\tau H}$, such as the one in \eqref{strang}, the analysis can start with the one-step error, $V^{1/r} - U^{1/r} = \CO{\tau^{p+1}}$. Therefore, in the estimate \eqref{eq: t3}, one can improve $p$ to $p+1.$ 
\end{remark}

By setting $\delta= C \tau^p$, this estimate coincides with that in Theorem \ref{eq: thm2}. 

By neglecting the constants in the logarithmic function, and the fact that $\tau=1/r$, we can further estimate $t$ as follows,
\[
 t \approx n  + \log \frac{1}{ \Delta\!E^2 r^{2p} \eps}. 
\]
This can be used as a guideline to choose the step number $r,$
\begin{equation}\label{r-scale}
 r \ge   \frac{e^{ \frac{n-t}{2p}}}{\Delta\!E^{\frac1p} \eps^{\frac1{2p}}}.
\end{equation}

\section{Random Unitary Approximations}

Another interesting framework in Hamiltonian simulations is to approximate the unitary operator  by a random sampling. 
An implementation of this algorithm, known as the QDRIFT method \cite{campbell2019random,ouyang2020compilation}, can be described as follows. First we construct a discrete probability for sampling individual Hamiltonian terms,
\begin{equation}\label{eq: p-lam}
    p_\ell  = \frac{\norm{h_\ell}}{\lambda}, \ell=1,2,\cdots, L.  \quad 
\lambda := \sum_{\ell=1}^L  \norm{h_\ell}.
\end{equation}
 Then, for each step of the approximation, one picks a Hamiltonian term $h_\ell$ with probability $p_\ell$, and let $W_j= \exp(i\tau h_\ell/p_\ell). $   Notice that the index $\ell$ is random. The result of repeating this approximation $r$ times is a product form, given by,
\begin{equation}\label{eq: rtrott}
    U \approx V_r, \quad V_r=W_r W_{r-1} \cdots W_1. 
\end{equation}
By using a Taylor expansion, Campbell \cite{campbell2019random} showed that on average, this approximation is consistent, {in the sense that the one step error is at most $\mathcal{O}(\tau^2)$ \cite{deuflhard2002scientific}.} The most appealing aspect  of this type of methods is that at each step, only {\it one} Hamiltonian term needs to be implemented. On the other hand, it is also possible to draw several Hamiltonian terms at each step, or treat some dominant terms always in a deterministic manner. These generalizations can improve the statistical accuracy by reducing the variance \cite{jin2021partially}. 

To abstract out the specific details of the sampling methods, we simply consider a formalism where $\{W_\gamma\}_{\gamma \ge 0}$ are random, matrix-valued variables.  In addition, they are independent  and identically distributed (i.i.d.) unitary matrices. 

To ensure that such $V_r$ is consistent with $U$ on average,  we make the first assumption that, for some integer $q>1$, the one-step error is bounded by,
\begin{equation}\label{eq: A1}
  \norm{(\overline W  - U^{1/r} ) \ket{\psi} } \leq C \lambda^2 \tau^q, \;      \overline{W}:= \bbe\left[ W \right].    
\end{equation}
The constant $C$ is independent of $\tau.$ The order condition ensures the statistical consistency of the random approximation. 
For instance, the analysis of \cite{chen2020quantum} revealed an error bound of $ \CO{\lambda^2\tau^2}$ with $\lambda$ from \eqref{eq: p-lam}. In particular, $\lambda^2$ represents the variance of the {importance sampling} of the Hamiltonian terms \cite{jin2021partially}. In this case, we can set  $q=2.$ 
Due to independence, 
\begin{equation}\label{V-mean}
\overline{V}= \bbe\left[ V_r \right] =\bbe\left[ W \right]^r
\end{equation}
 Thus, a simple extension of this error bound is, for each $k\in \mathbb{N},$
\begin{equation}
  \norm{(\overline{V}^{2^k}  - U^{2^k} ) \ket{\psi} } \leq C r 2^ k \lambda^2 \tau^q. 
  \end{equation}
However, this error grows rather quickly, and we will rely on a concentration inequality instead, as shown below.

{
We now apply \eqref{eq: rtrott} to the phase estimation method, and in particular,  for each qubit in the first quantum register, we replace {$U^{2^k}$ by $Q_k $ which is defined as,
\begin{equation}\label{Qk}
 Q_k = \prod_{j=1}^{r2^k}{W}_{j}.
\end{equation}
{The basic idea is to separate the matrix $V_r$ into its mean $\overline{V}_r$ \eqref{V-mean}. Following \eqref{eq: V-SD}, one can express $\overline{V}_r$ in a spectral decomposition form with eigenvectors $\ket{\chi_m}$ associated with 
phase values $E_m$. In light of the mean error bound \eqref{eq: A1}, the eigenvectors of $\overline{V}_r$, according to Lemma \ref{lmm2}, have overlap with those of $U$.   
After applying the controlled-$U$ gates, we find that,
\begin{align}\label{eq: phi2-rand}
\ket{\phi_2}= &  \ket{\widehat{\phi}_2} + \ket{\widetilde{\phi}_2}, \\ 
\ket{\widehat{\phi}_2} = & \dfrac{1}{2^{t/2}}   \sum_{k=0}^{2^t-1}   \sum_{m=1}^d  e^{2\pi i k E_m } \braket{\psi_m}{\psi} \ket{k}.
\end{align}}
where we have denoted the two terms on the right hand side by $\ket{\widehat{\phi}_2}$ and $ \ket{\widetilde{\phi}_2}$, respectively. In particular, $\ket{\widehat{\phi}_2}$ is the quantum state when the $U$ is replaced by $\overline{V}$ in the PEA.  Similar to the Trotter approximation, this is a deterministic approximation.  In light of Lemma \ref{lmm2} and Eq. \eqref{eq: A1}, we still have the estimate for the overlap (when $E_m$ is the closest to the exact phase value), 
\begin{equation}\label{eq: over1}
 \abs{\braket{\chi_m }{\psi} }^2 \geq 1 -  \Big( \frac{ C  \lambda^2 \tau^{q}}{2\Delta E} \Big)^2.
\end{equation}
This part is the same as the analysis in Theorem \ref{thm3}.

On the other hand, $ \ket{\widetilde{\phi}_2}$ can be regarded as a fluctuation. Due to the independence of the random matrices, $ \ket{\widetilde{\phi}_2}$ has zero mean. More importantly, the product of the random unitary matrices has a tendency to concentrate around the mean. Toward this end,
we define $f= 2 \text{Re} \braket{\widehat{\phi}_2}{\widetilde{\phi}_2}$.  {Since $\ket{\widehat{\phi}_2}$ is deterministic, one has $\mathbb{E}[f]=0$. 
Therefore a concentration inequality would provide the probability of $f$ being close to zero. Furthermore,  
the fact that both  $\ket{\widehat{\phi}_2}$ and $\ket{{\phi}_2}$ are unit vectors implies that  $\braket{\widetilde{\phi}_2}= -{f},$ 
and this can be used to estimate the success probability  of obtaining  $ \ket{\widehat{\phi}_2}$ in \eqref{eq: phi2-rand}. To proceed in this direction, we rewrite \eqref{eq: phi2-rand} as,
\begin{equation}\label{eta}
\ket{{\phi}_2}=   \eta \ket{\widehat{\phi}_2} + \ket{\widehat{\phi}_2^\perp}, \quad \eta= \braket{\widehat{\phi}_2}{{\phi}_2}
\end{equation}
where $\ket{\widehat{\phi}_2^\perp}$ is regarded as a `garbage' state. In order to apply the analysis in Theorem \ref{thm3} to $\ket{\widehat{\phi}_2} $ and obtain a similar result, we enforce the event $\abs{\eta}^2 > 1 - \epsilon,$ with large probability. To fulfill this condition, we notice from \eqref{eq: phi2-rand} and  \eqref{eta} that $$\abs{\eta} = \abs{1 + \braket{\widehat{\phi}_2}{\widetilde{\phi}_2}}  \geq 1 - \abs{\braket{\widehat{\phi}_2}{\widetilde{\phi}_2}} \geq 1 - \sqrt{\braket{\widetilde{\phi}_2}{\widetilde{\phi}_2}} = 1 - \abs{f}^{1/2}. $$ 
In light of this, we will analyze the probability that,
\begin{equation}\label{f-ieq}
  |f|< \epsilon^2/4,
\end{equation}
which implies that $\abs{\eta}^2 > 1 - \epsilon.$

Overall, $f$ can be regarded as a function of the random matrices $W_1, W_2, \cdots, W_{r2^{t-1}}$. We first notice that, in the $j$th application of the random algorithm $W_j$, if a different Hamiltonian is selected, then the error is at most \cite{jin2021partially},
\[ \norm{W_j - W_j'} \leq 2 \lambda \tau,  \; \; \forall  1\leq j \leq r 2^{t-1}, \]
where we have used $W'$ for the unitary operator that corresponds to a different sample of the Hamiltonian terms. Due to the unitary property, this implies that for each $0 \leq k < t$, the controlled-$U$ operator has a perturbation with the same bound,
\begin{equation}
\norm{W_{r2^k} \cdots W_{j+1} W_j W_{j-1}\cdots W_1 -W_{r2^k}\cdots W_{j+1} W_j' W_{j-1} \cdots W_1} \leq 2 \lambda \tau.\end{equation}

 If we further denote the resulting function value (by substituting one $W_j$ in the $k$th step of the controlled-$U$ gate)  by $f'$, then we have the following rough bound by using triangle inequality and $\tau=1/r$,
\[ |f - f'| \leq {2 \lambda \tau } \Rightarrow  \sum_{j=1}^{r 2^{t-1}}  |f - f_j'|^2 \leq  { 2^{t+1}
 r \lambda^2 \tau^2 } = { 2^{t+1}
  \lambda^2/r }.  \]}
This observation can be used to invoke the McDiarmid  inequality \cite{mcdiarmid1989method}, showing that, for any $\nu>0$,
\begin{equation}\label{eq: mcd}
\mathbb{P}\left( \abs{f}> \nu\right) < \exp \Big( -\frac{r \nu^2}{  2^{t+1} \lambda^2} \Big).
\end{equation}

To maintain an overall success probability of $1-\epsilon$, we set 
\[ \epsilon \geq 2 \exp \Big( -\frac{r \nu^2}{  2^{t+1} \lambda^2} \Big),\]
which yields a lower bound for $r$,
\begin{equation}\label{eq: r-bd}
r \geq \frac{2^{t+1} \lambda^2 \log \frac{2}{\epsilon}}{\nu^2}.
\end{equation} 
 Such type of lower bounds have also been  obtained in \cite{Chen2021,jin2021partially} as a gate complexity estimate. The term $2^t$ can be viewed as the time duration of the Hamiltonian simulation.

{
To apply the inequality \eqref{eq: mcd} to \eqref{f-ieq}, we set $\nu= \epsilon^2/4,$ which changes the bound \eqref{eq: r-bd} to,
\begin{equation}\label{eq: r-bd'}
r \geq \frac{2^{t+5} \lambda^2 \log \frac{2}{\epsilon}}{\epsilon^4}.
\end{equation} 
 Thus increasing $t$ also has an effect on the number of steps $r$. The scaling $\epsilon^{-4}$ seems to be more restrictive than the scaling from the deterministic method \eqref{r-scale}. On the other hand, the random algorithm only involves one Hamiltonian term at each time step. }

 Collecting these results, we have the estimate on $t$.
\begin{theorem}
  The quantum algorithm using random unitaries returns an estimation with {error less or equal to $2^{-n}$ and}  with success probability at least $1-\eps$ using $t$ qubits with $t$ in the following range: 
    \begin{equation}\label{eq: t4}
 \log \frac{r\epsilon^4}{\lambda^2\log \frac{2}\epsilon} - 5 \geq    t \geq n + \log \big(2 + \frac{ C^2 }{2  \Delta\!E^2 r^{2q}\eps} \big). 
   \end{equation}
   
\end{theorem}
A finite range for $t$ can always be found by increasing $r$,
since the upper bound is an increasing function of $r$, and the lower bound decreases with $r$.  
}

\section*{Acknowledgement} {The author's research on quantum algorithms is supported by the National Science Foundation Grants DMS-2111221. }

\bigskip

\bibliographystyle{plain}
\bibliography{qcomp,PEA,trotter}

\begin{thebibliography}{10}

\bibitem{An2021}
Dong An, Di~Fang, and Lin Lin.
\newblock Time-dependent hamiltonian simulation of highly oscillatory dynamics.
\newblock {\em arXiv preprint arXiv:2111.03103}, 2021.

\bibitem{an2021time}
Dong An, Di~Fang, and Lin Lin.
\newblock Time-dependent unbounded hamiltonian simulation with vector norm
  scaling.
\newblock {\em Quantum}, 5:459, 2021.

\bibitem{babbush2018encoding}
Ryan Babbush, Craig Gidney, Dominic~W Berry, Nathan Wiebe, Jarrod McClean,
  Alexandru Paler, Austin Fowler, and Hartmut Neven.
\newblock Encoding electronic spectra in quantum circuits with linear t
  complexity.
\newblock {\em Physical Review X}, 8(4):041015, 2018.

\bibitem{babbush_chemical_2015}
Ryan Babbush, Jarrod McClean, Dave Wecker, Alán Aspuru-Guzik, and Nathan
  Wiebe.
\newblock Chemical basis of {Trotter}-{Suzuki} errors in quantum chemistry
  simulation.
\newblock {\em Physical Review A}, 91(2), February 2015.

\bibitem{bauman2020toward}
Nicholas~P Bauman, Hongbin Liu, Eric~J Bylaska, Sriram Krishnamoorthy,
  Guang~Hao Low, Christopher~E Granade, Nathan Wiebe, Nathan~A Baker, Bo~Peng,
  Martin Roetteler, et~al.
\newblock Toward quantum computing for high-energy excited states in molecular
  systems: quantum phase estimations of core-level states.
\newblock {\em Journal of Chemical Theory and Computation}, 17(1):201--210,
  2020.

\bibitem{berni2015ab}
Adriano~A Berni, Tobias Gehring, Bo~M Nielsen, Vitus H{\"a}ndchen, Matteo~GA
  Paris, and Ulrik~L Andersen.
\newblock Ab initio quantum-enhanced optical phase estimation using real-time
  feedback control.
\newblock {\em Nature Photonics}, 9(9):577--581, 2015.

\bibitem{Berry2015}
Dominic~W Berry, Andrew~M Childs, and Robin Kothari.
\newblock Hamiltonian simulation with nearly optimal dependence on all
  parameters.
\newblock In {\em 2015 IEEE 56th Annual Symposium on Foundations of Computer
  Science}, pages 792--809. IEEE, 2015.

\bibitem{Berry2020}
Dominic~W Berry, Andrew~M Childs, Yuan Su, Xin Wang, and Nathan Wiebe.
\newblock Time-dependent hamiltonian simulation with $ l^1$-norm scaling.
\newblock {\em Quantum}, 4:254, 2020.

\bibitem{bian2019quantum}
Teng Bian, Daniel Murphy, Rongxin Xia, Ammar Daskin, and Sabre Kais.
\newblock Quantum computing methods for electronic states of the water
  molecule.
\newblock {\em Molecular Physics}, 117(15-16):2069--2082, 2019.

\bibitem{brassard2011optimal}
Gilles Brassard, Frederic Dupuis, Sebastien Gambs, and Alain Tapp.
\newblock An optimal quantum algorithm to approximate the mean and its
  application for approximating the median of a set of points over an arbitrary
  distance.
\newblock {\em arXiv preprint arXiv:1106.4267}, 2011.

\bibitem{brassard2002quantum}
Gilles Brassard, Peter Hoyer, Michele Mosca, and Alain Tapp.
\newblock Quantum amplitude amplification and estimation.
\newblock {\em Contemporary Mathematics}, 305:53--74, 2002.

\bibitem{campbell2019random}
Earl Campbell.
\newblock Random compiler for fast {Hamiltonian} simulation.
\newblock {\em Physical review letters}, 123(7):070503, 2019.

\bibitem{campbell2021early}
Earl~T Campbell.
\newblock Early fault-tolerant simulations of the hubbard model.
\newblock {\em Quantum Science and Technology}, 7(1):015007, 2021.

\bibitem{casares2021t}
PAM Casares, Roberto Campos, and MA~Martin-Delgado.
\newblock {T-Fermion}: A {non-Clifford} gate cost assessment library of quantum
  phase estimation algorithms for quantum chemistry.
\newblock {\em arXiv preprint arXiv:2110.05899}, 2021.

\bibitem{chappell2011precise}
James~M Chappell, Max~A Lohe, Lorenz Von~Smekal, Azhar Iqbal, and Derek Abbott.
\newblock A precise error bound for quantum phase estimation.
\newblock {\em Plos one}, 6(5):e19663, 2011.

\bibitem{chen2020quantum}
Chi-Fang Chen, Hsin-Yuan Huang, Richard Kueng, and Joel~A Tropp.
\newblock Quantum simulation via randomized product formulas: Low gate
  complexity with accuracy guarantees.
\newblock {\em arXiv preprint arXiv:2008.11751}, 2020.

\bibitem{Chen2021}
Yi-Hsiang Chen, Amir Kalev, and Itay Hen.
\newblock Quantum algorithm for time-dependent hamiltonian simulation by
  permutation expansion.
\newblock {\em PRX Quantum}, 2(3):030342, 2021.

\bibitem{childs2019faster}
Andrew~M Childs, Aaron Ostrander, and Yuan Su.
\newblock Faster quantum simulation by randomization.
\newblock {\em Quantum}, 3:182, 2019.

\bibitem{childs2019nearly}
Andrew~M Childs and Yuan Su.
\newblock Nearly optimal lattice simulation by product formulas.
\newblock {\em Physical review letters}, 123(5):050503, 2019.

\bibitem{childs2021theory}
Andrew~M Childs, Yuan Su, Minh~C Tran, Nathan Wiebe, and Shuchen Zhu.
\newblock Theory of {Trotter} error with commutator scaling.
\newblock {\em Physical Review X}, 11(1):011020, 2021.

\bibitem{cleve1998quantum}
Richard Cleve, Artur Ekert, Chiara Macchiavello, and Michele Mosca.
\newblock Quantum algorithms revisited.
\newblock {\em Proceedings of the Royal Society of London. Series A:
  Mathematical, Physical and Engineering Sciences}, 454(1969):339--354, 1998.

\bibitem{daskin2018direct}
Ammar Daskin and Sabre Kais.
\newblock Direct application of the phase estimation algorithm to find the
  eigenvalues of the {Hamiltonians}.
\newblock {\em Chemical Physics}, 514:87--94, 2018.

\bibitem{deuflhard2002scientific}
Peter Deuflhard and Folkmar Bornemann.
\newblock {\em Scientific computing with ordinary differential equations},
  volume~42.
\newblock Springer Science \& Business Media, 2002.

\bibitem{faehrmann2021randomizing}
Paul~K Faehrmann, Mark Steudtner, Richard Kueng, Maria Kieferova, and Jens
  Eisert.
\newblock Randomizing multi-product formulas for improved {Hamiltonian}
  simulation.
\newblock {\em arXiv preprint arXiv:2101.07808}, 2021.

\bibitem{harrow2009quantum}
Aram~W. Harrow, Avinatan Hassidim, and Seth Lloyd.
\newblock Quantum algorithm for solving linear systems of equations.
\newblock {\em Physical Review Letters}, 103(15):150502, October 2009.
\newblock arXiv: 0811.3171.

\bibitem{hastings2014improving}
Matthew~B Hastings, Dave Wecker, Bela Bauer, and Matthias Troyer.
\newblock Improving quantum algorithms for quantum chemistry.
\newblock {\em arXiv preprint arXiv:1403.1539}, 2014.

\bibitem{jin2021partially}
Shi Jin and Xiantao Li.
\newblock A partially random trotter algorithm for quantum hamiltonian
  simulations.
\newblock {\em arXiv preprint arXiv:2109.07987}, 2021.

\bibitem{kitaev1995quantum}
A~Yu Kitaev.
\newblock Quantum measurements and the abelian stabilizer problem.
\newblock {\em arXiv preprint quant-ph/9511026}, 1995.

\bibitem{Kivlichan2020}
Ian~D Kivlichan, Craig Gidney, Dominic~W Berry, Nathan Wiebe, Jarrod McClean,
  Wei Sun, Zhang Jiang, Nicholas Rubin, Austin Fowler, Al{\'a}n Aspuru-Guzik,
  et~al.
\newblock Improved fault-tolerant quantum simulation of condensed-phase
  correlated electrons via trotterization.
\newblock {\em Quantum}, 4:296, 2020.

\bibitem{Lin2016}
Lin Lin, Yousef Saad, and Chao Yang.
\newblock Approximating spectral densities of large matrices.
\newblock {\em SIAM review}, 58(1):34--65, 2016.

\bibitem{lin2020near}
Lin Lin and Yu~Tong.
\newblock Near-optimal ground state preparation.
\newblock {\em Quantum}, 4:372, 2020.

\bibitem{lin2022heisenberg}
Lin Lin and Yu~Tong.
\newblock Heisenberg-limited ground-state energy estimation for early
  fault-tolerant quantum computers.
\newblock {\em PRX Quantum}, 3(1):010318, 2022.

\bibitem{Lloyd1996}
Seth Lloyd.
\newblock Universal quantum simulators.
\newblock {\em Science}, 273(5278):1073--1078, 1996.

\bibitem{Low2017}
Guang~Hao Low and Isaac~L Chuang.
\newblock Optimal hamiltonian simulation by quantum signal processing.
\newblock {\em Physical review letters}, 118(1):010501, 2017.

\bibitem{Low2019}
Guang~Hao Low and Isaac~L Chuang.
\newblock Hamiltonian simulation by qubitization.
\newblock {\em Quantum}, 3:163, 2019.

\bibitem{Low2018}
Guang~Hao Low and Nathan Wiebe.
\newblock Hamiltonian simulation in the interaction picture.
\newblock {\em arXiv preprint arXiv:1805.00675}, 2018.

\bibitem{mcdiarmid1989method}
Colin McDiarmid.
\newblock On the method of bounded differences.
\newblock {\em Surveys in combinatorics}, 141(1):148--188, 1989.

\bibitem{moore2021statistical}
Alexandria~J Moore, Yuchen Wang, Zixuan Hu, Sabre Kais, and Andrew~M Weiner.
\newblock Statistical approach to quantum phase estimation.
\newblock {\em arXiv preprint arXiv:2104.10285}, 2021.

\bibitem{mosca1999quantum}
Michele Mosca.
\newblock {\em Quantum computer algorithms}.
\newblock PhD thesis, University of Oxford. 1999., 1999.

\bibitem{nielsen2002quantum}
Michael~A Nielsen and Isaac Chuang.
\newblock Quantum computation and quantum information, 2002.

\bibitem{ouyang2020compilation}
Yingkai Ouyang, David~R White, and Earl~T Campbell.
\newblock Compilation by stochastic {Hamiltonian} sparsification.
\newblock {\em Quantum}, 4:235, 2020.

\bibitem{poulin_trotter_2014}
David Poulin, M.~B. Hastings, Dave Wecker, Nathan Wiebe, Andrew~C. Doherty, and
  Matthias Troyer.
\newblock The {Trotter} {Step} {Size} {Required} for {Accurate} {Quantum}
  {Simulation} of {Quantum} {Chemistry}.
\newblock {\em arXiv:1406.4920 [quant-ph]}, June 2014.
\newblock arXiv: 1406.4920.

\bibitem{poulin2018quantum}
David Poulin, Alexei Kitaev, Damian~S Steiger, Matthew~B Hastings, and Matthias
  Troyer.
\newblock Quantum algorithm for spectral measurement with a lower gate count.
\newblock {\em Physical review letters}, 121(1):010501, 2018.

\bibitem{rendon2021effects}
Gumaro Rendon, Taku Izubuchi, and Yuta Kikuchi.
\newblock Effects of cosine tapering window on quantum phase estimation.
\newblock {\em arXiv preprint arXiv:2110.09590}, 2021.

\bibitem{saad2011numerical}
Yousef Saad.
\newblock {\em Numerical methods for large eigenvalue problems: revised
  edition}.
\newblock SIAM, 2011.

\bibitem{shor1999polynomial}
Peter~W Shor.
\newblock Polynomial-time algorithms for prime factorization and discrete
  logarithms on a quantum computer.
\newblock {\em SIAM review}, 41(2):303--332, 1999.

\bibitem{Su2021}
Yuan Su, Dominic~W Berry, Nathan Wiebe, Nicholas Rubin, and Ryan Babbush.
\newblock Fault-tolerant quantum simulations of chemistry in first
  quantization.
\newblock {\em PRX Quantum}, 2(4):040332, 2021.

\bibitem{tansuwannont2019quantum}
Theerapat Tansuwannont, Surachate Limkumnerd, Sujin Suwanna, and Pruet
  Kalasuwan.
\newblock Quantum phase estimation algorithm for finding polynomial roots.
\newblock {\em Open Physics}, 17(1):839--849, 2019.

\bibitem{temme2011quantum}
Kristan Temme, Tobias~J Osborne, Karl~G Vollbrecht, David Poulin, and Frank
  Verstraete.
\newblock Quantum metropolis sampling.
\newblock {\em Nature}, 471(7336):87--90, 2011.

\bibitem{wan2021randomized}
Kianna Wan, Mario Berta, and Earl~T Campbell.
\newblock A randomized quantum algorithm for statistical phase estimation.
\newblock {\em arXiv preprint arXiv:2110.12071}, 2021.

\bibitem{Watkins2022}
Jacob Watkins, Nathan Wiebe, Alessandro Roggero, and Dean Lee.
\newblock Time-dependent hamiltonian simulation using discrete clock
  constructions.
\newblock {\em arXiv preprint arXiv:2203.11353}, 2022.

\bibitem{whitfield2011simulation}
James~D Whitfield, Jacob Biamonte, and Al{\'a}n Aspuru-Guzik.
\newblock Simulation of electronic structure hamiltonians using quantum
  computers.
\newblock {\em Molecular Physics}, 109(5):735--750, 2011.

\bibitem{wille2019ibm}
Robert Wille, Rod Van~Meter, and Yehuda Naveh.
\newblock {IBM’s} {Qiskit} tool chain: Working with and developing for real
  quantum computers.
\newblock In {\em 2019 Design, Automation \& Test in Europe Conference \&
  Exhibition (DATE)}, pages 1234--1240. IEEE, 2019.

\bibitem{xiao2019continuous}
Tailong Xiao, Jingzheng Huang, Jianping Fan, and Guihua Zeng.
\newblock Continuous-variable quantum phase estimation based on machine
  learning.
\newblock {\em Scientific reports}, 9(1):1--13, 2019.

\end{thebibliography}

\end{document}